\documentclass[letterpaper, 10 pt, conference]{ieeeconf}  %

\IEEEoverridecommandlockouts                              %

\usepackage{times}

\usepackage{graphics} %
\usepackage{times} %
\usepackage{amsmath, amssymb} %
\usepackage{booktabs} %
\usepackage[group-digits=false]{siunitx} %
\usepackage{graphicx} %
\usepackage{hyperref} 
\usepackage[capitalise]{cleveref}
\usepackage{algorithm}
\usepackage{algpseudocode}
\usepackage{mathtools}
\usepackage{tikz}
\usepackage{color, soul}
\usetikzlibrary{patterns,shapes,scopes,arrows,positioning,calc,decorations.pathmorphing,math,angles,quotes,decorations.markings}
\usepackage{diagbox}

\DeclareMathOperator*{\argmin}{\arg\!\min}

\newtheorem{lemma}{Lemma}

\newrobustcmd*{\bftabnum}{%
  \bfseries
  \sisetup{output-decimal-marker={\textmd{.}}}%
}
\begin{document}

\title{\LARGE \bf Does Bilevel Optimization Result in More Competitive Racing Behavior?
}

\author{Andrew Cinar$^{1}$ and Forrest Laine$^{2}$%
\thanks{*This work was not supported by any organization}%
\thanks{$^{1}$Andrew Cinar is with the Department of Mechanical Engineering, Vanderbilt University,
        Nashville, TN 37235, USA
        {\tt\small al.cinar@vanderbilt.edu}}%
\thanks{$^{2}$Forrest Laine is with the Department of Computer Science, Vanderbilt University,
        Nashville, TN 37235, USA
        {\tt\small forrest.laine@vanderbilt.edu}}%
}

\mbox{\begin{minipage}[c]{\textwidth}
 \begin{center}
© 2025 IEEE.  Personal use of this material is permitted.  Permission from IEEE must be obtained for all other uses, in any current or future media, including reprinting/republishing this material for advertising or promotional purposes, creating new collective works, for resale or redistribution to servers or lists, or reuse of any copyrighted component of this work in other works. \\
\hfill \break
Accepted version.
 \end{center}
\end{minipage}
}
\thispagestyle{empty}
\newpage

\maketitle
\begin{abstract}
Two-vehicle racing is natural example of a competitive dynamic game. As with most dynamic games, there are many ways in which the underlying solution concept can be structured, resulting in different equilibrium concepts. The assumed solution concept influences the behaviors of two interacting players in racing. For example, blocking behavior emerges naturally in leader-follower play, but to achieve this in Nash play the costs would have to be chosen specifically to trigger this behavior. In this work, we develop a novel model for competitive two-player vehicle racing, represented as an equilibrium problem, complete with simplified aerodynamic drag and drafting effects, as well as position-dependent collision-avoidance responsibility. We use our model to explore how different solution concepts affect competitiveness. We develop a solution for bilevel optimization problems, enabling a large-scale empirical study comparing bilevel strategies (either as leader or follower), Nash equilibrium strategy and a single-player constant velocity baseline. We find the choice of strategies significantly affects competitive performance and safety. %

\end{abstract}

\IEEEpeerreviewmaketitle

\section{Introduction}
\label{sec:intro}
Game-theoretic motion planning has recently emerged as a promising approach for handling  multi-agent interactions in a principled manner, with many works exploring related ideas. However, some questions remain about the best way to formulate the trajectory games played between agents. The simplest and most common approach is to cast the interaction as a Generalized Nash Equilibrium Problem \cite{le2022algames,burger_interaction-aware_2022,ji_lane-merging_2021, williams_best_2018}, but this formulation has known deficiencies. In Nash equilibrium, players cannot anticipate how others will respond to changes in their actions. 
Some researchers suggest posing trajectory interactions as repeated games and solving for Generalized Feedback Nash Equilibria \cite{laine2023computation, fridovich2020efficient, laine2021multi, wang_game-theoretic_2021} to capture the intelligent reasoning lacking in static Nash equilibria. However, feedback equilibria are extremely difficult to solve, and the strategic benefit of using a more elaborate solution concept is unclear. 

Bilevel (Stackelberg) games lie between repeated and static games on the solution concept spectrum, with one player as the leader, and the other(s) as followers. The leader in bilevel optimization can predict how followers' actions will change based on its own, but the converse is untrue. Some studies have explored bilevel optimization in two-player games, leveraging the leader-follower dynamic to produce competitive behaviors like blocking \cite{liniger_noncooperative_2020, hu_multi-leader-follower_2015, yoo_stackelberg_2020}. However, computing true Stackelberg equilibria in the full, constrained trajectory space is also challenging due to the inherent complexity. %

Even if such problems could be solved, which player should be deemed the leader and which the follower is not clear. These distinctions may seem baseless in real-world situations where the players do not coordinate their roles and make decisions simultaneously. However, these distinctions are not irrelevant for autonomous decision-making in multi-agent interactions. Different formulations lead to different strategic behaviors. %

We believe advanced solution concepts offer benefits, but there exists a literature gap exploring to what degree a benefit is admitted. This study aims to close that gap by performing an empirical analysis of benefits of bilevel reasoning compared to static reasoning and single-player optimization. We focus on relatively basic solution concepts, because even for these equilibria, the costs and benefits of choosing among them are not well understood. Furthermore, we believe Stackelberg equilibria are the simplest formulation computationally that can reason about action and reaction. We believe our results can serve as a foundation for future comparisons of solution concepts.

\begin{figure}[t]
\centering
\includegraphics[width=0.65\columnwidth]{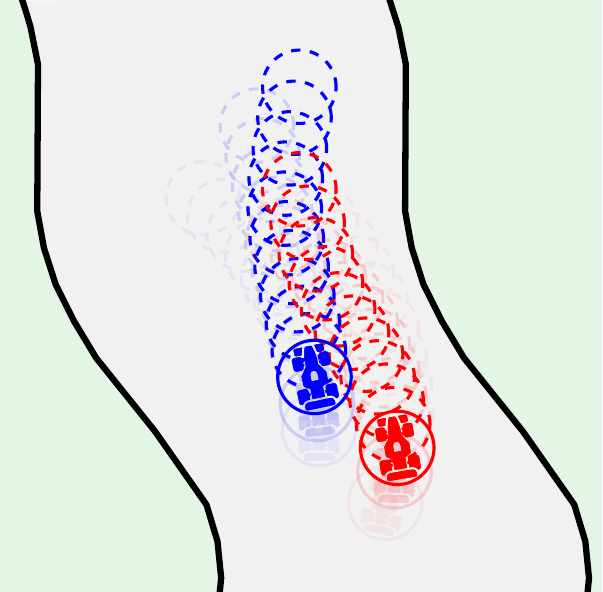}
\caption{Red (Player 2) Follower is blocked by Blue (Player 1) Leader while attempting to pass in bilevel L-F competition.}
\label{fig:bilevel_example}
\end{figure}

We focus on two-agent racing, but the results should be informative for all roboticists developing autonomous systems in multi-agent settings. Racing is a challenging, interactive domain, and so it is a good setting to study various interaction models. Our results are intended to guide appropriate information structure selection in the general class of finite dynamic games. Racing interactions can apply to autonomous driving or general human-robot interactions, where the agents must coordinate to maintain safety (like avoiding collisions) while achieving satisfactory individual performance, like minimizing control effort.

The contributions we make in this work are the following: 
\begin{enumerate}
    \item We implement a solver for nonlinear, constrained bilevel optimization.
    \item We derive a novel two-player racing model to capture key competitive aspects of racing, incorporating position-dependent collision avoidance responsibility, aerodynamic drag, drafting effects, nonlinear dynamics and curving tracks.\footnote{The code is available at \href{https://github.com/VAMPIR-Lab/EpecRacing}{https://github.com/VAMPIR-Lab/EpecRacing}.}
    \item We present a large scale empirical study analyzing the performance and robustness of 16 types of two-player racing competitions with different solution concept combinations of: (1) Single-player optimization, (2) generalized Nash equilibrium, (3) bilevel equilibrium as leader, or (4) bilevel equilibrium as follower.
\end{enumerate}

\section{Game Theory Background}
\label{sec:game_theory}
This work is primarily concerned with two formulations of a two-player mathematical game. Consider two decision-makers, player $1$ and player $2$. Each player $i$ has a set of decision variables $x_i \in \mathbb{R}^{n_i}$, termed the private decision variables. The joint set of decision variables is $x \in \mathbb{R}^n := [x_1^\intercal, x_2^\intercal]^\intercal$, where $n = n_1 + n_2$. We will use the shorthand $(x_i,x_{-i}) \equiv x$ referring both players, and use index $-i$ to refer to player $2$ when $i=1$ and vice-versa. 
Each player's preference is characterized by a continuous cost function $f_i : \mathbb{R}^{n} \to \mathbb{R}$, and a feasible region defined by $C_i := \{x\in\mathbb{R}^n : g_i(x) \ge 0\}$, where $g_i : \mathbb{R}^{n} \to \mathbb{R}^{m_i}$ is a vector-valued constraint function. We require that the functions $f_i$ and $g_i$ are twice-differentiable. 

For player $i$, the standard solution graph for their decision problem is defined as,
\begin{equation} \label{eq:solution_graph}
    S_i := \left\{ \begin{aligned} x^* \in \mathbb{R}^n \ : \ x^*_i \in &\argmin_{x_i}  && f_i\left(x_i, x^*_{-i}\right) \\
        & \ \ \ \ \text{s.t.} && \left(x_i, x^*_{-i}\right) \in C_i \end{aligned} \right\}.
\end{equation}
In this work, the minimization in \cref{eq:solution_graph} and other problems are assumed to be local unless specified otherwise. Solution graph $S_i$ is the set of all points $x_i^*$ which are local optimizers for the problem parameterized by $x_{-i}^*$.
A point $x^* \in \mathbb{R}^{n}$ is called a local Generalized Nash Equilibrium Point if 
\begin{equation} \label{eq:gnep}
\begin{aligned}
    x^* \in S_1 \cap S_2.
\end{aligned}
\end{equation}
A Generalized Nash Equilibrium Problem computes a point $x^*$ such that it satisfies \cref{eq:gnep} \cite{nash1951non,debreu1952social,arrow1954existence,facchinei2007generalized}. For brevity, we often drop the label ``generalized''.

For player $i$, we define the bilevel solution graph: 
\begin{equation}
    \label{eq:solution_graph_bilevel}
    B_i := \left\{ \begin{aligned} 
        x^* \in \mathbb{R}^n \ : \ x^* \in &\argmin_{x}  && f_i\left(x\right) \\
        & \ \ \ \ \text{s.t.} && x \in C_i,\, x \in S_{-i}
        \end{aligned} \right\}
\end{equation}
The bilevel solution graph for player $i$, the \emph{leader}, is the set of all local optimizer points for the bilevel problem, where player $-i$ is the \emph{follower}. A bilevel equilibrium problem is the problem of finding the points $x^* \in B_i$, termed the bilevel equilibrium.

In duopoly theory from economics, Generalized Nash Equilibrium Problems are analogous to finding equilibrium points in Cournot competition \cite{cournot1927researches}, and bilevel optimization problems are analogous to equilibrium points for a Stackelberg competition \cite{von2010market}. It has been shown that Stackelberg duopolies generally result in an increase in total welfare compared to Cournot duopolies \cite{huck2001stackelberg,daughety1990beneficial}. We investigate whether a similar results hold for Bilevel or Nash competitions in the racing domain.

\section{Computing Solutions to Equilibrium Problems}
\label{sec:bilevel_solver}
In this section we present methodologies for computing Nash and bilevel equilibrium points. 

\subsection{Nash Equilibrium Points}\label{sec:nash}
We assume a suitable constraint qualification is satisfied for the problems \cref{eq:solution_graph}. We invoke the Karush-Kuhn-Tucker (KKT) theorem \cite{karush1939minima,kuhntucker} to write the first-order necessary conditions which must hold for every point $x^* \in S_i$:
\begin{equation} \label{eq:kkt}
    \begin{aligned}
        x^* \in S_i &\implies  \exists \lambda_i \in \mathbb{R}^{m_i} : (x^*, \lambda_i) \in \overline{S_i},
    \end{aligned}
\end{equation}
\begin{equation}\label{eq:stationary_graph}
        \overline{S_i} := \left\{     \begin{aligned} (x^*, \lambda_i) : &\ \nabla_{x_i}f_i(x^*) - (\nabla_{x_i}g_i(x^*))^\intercal \lambda_i = 0 \\
        &\ 0 \le g_i(x^*) \perp \lambda_i \ge 0. \end{aligned}\right\}.
\end{equation}
Here, $0 \le a \perp b \ge 0$ for $a, b\in\mathbb{R}^m$ means $a_i \ge 0$, $b_i \ge 0$, and $a_ib_i = 0 \ \forall i \in \{1,...,m\}$.
Combining \cref{eq:kkt,eq:stationary_graph,eq:gnep}, we have the following:
\begin{equation} \label{eq:mcp}
    \begin{aligned}
        x^* \in S_1 \cap S_2 \implies &\exists \lambda \in \mathbb{R}^{m_1+m_2}:\,
        F(z) \perp l \le z \le u,
    \end{aligned}
\end{equation}
\begin{equation}
    \begin{aligned}
        z &:= \begin{bmatrix} x^* \\ \lambda \end{bmatrix}, & F(z) &:= \begin{bmatrix} 
            \nabla_{x_1}f_1(x^*) - (\nabla_{x_1}g_1(x^*))^\intercal \lambda_1 \\
            \nabla_{x_2}f_2(x^*) - (\nabla_{x_2}g_2(x^*))^\intercal \lambda_2 \\
            g_1(x^*) \\
            g_2(x^*)
        \end{bmatrix}, \\
        l  &:= \begin{bmatrix} 0 \\ 0 \end{bmatrix},  & u &:= \begin{bmatrix} 0 \\ \infty \end{bmatrix}.
    \end{aligned}
\end{equation}
The first and second terms in $l$ and $u$ are vectors of dimension $n$ and $m = m_1+m_2$, respectively, and \cref{eq:mcp} means: 
\begin{equation} \label{eq:comp_def} \forall i \in \{1, \cdots, n + m\},
    \begin{cases}
        F_i(z) = 0, l_i < z_i < u_i, &\ \mathrm{or} \\
        F_i(z) > 0, l_i = z_i, &\ \mathrm{or} \\
        F_i(z) < 0, z_i = u_i.
    \end{cases}
\end{equation}

It is seen then that the conditions \cref{eq:mcp} form a mixed complementarity problem \cite{facchinei2003finite}. In this work we use the PATH Solver \cite{dirkse1995path} to find solutions $z^*$ to \cref{eq:mcp}. Any such point is not necessarily a Nash equilibrium point \cref{eq:gnep}, since in general the necessary conditions \cref{eq:kkt} are not sufficient for optimality. Nevertheless, points satisfying \cref{eq:mcp} are often used as a proxy for Nash equilibrium points \cite{facchinei2007generalized,facchinei_generalized_2009,dreves_nonsmooth_2011}, and we will do the same for the purposes of this work.

It is important to note that a Nash equilibrium, or even a solution to the mixed complementarity problem (\ref{eq:mcp}), may not exist. Existence of solution is generally hard to prove in games with constrained, non-convex optimization problems, as is the case here. However, we can reliably find equilibrium points  using the procedure described in this section for the games we consider in this work.

\subsection{Bilevel Equilibrium Points}\label{sec:bilevel}

In contrast to the necessary-condition approach to Nash equilibrium problems, we will define sufficient conditions for bilevel equilibrium points. First, note that $S_i \subset \{ x: \exists \lambda, (x,\lambda) \in \overline{S_i}\}$. The sets $B_i$ can be approximated by replacing the constraint involving $S_{-i}$ with the looser constraint:
\begin{equation}
    \label{eq:stationary_graph_bilevel}
    \overline{B_i} := \left\{ \begin{aligned} 
        x^* : \ \exists \lambda^*, (x^*,\lambda^*) \in &\argmin_{x,\lambda}  && f_i\left(x\right) \\
        & \ \ \ \ \text{s.t.} && x \in C_i \\
        & && \left(x,\lambda\right) \in \overline{S_{-i}}
        \end{aligned} \right\}.
\end{equation}
However, it can be seen that if $x^* \in S_{-i} \cap \overline{B_i}$, then it must also be that $x^* \in B_i$. Therefore, we pursue characterizing the points within $\overline{B_i}$. 
Note that the sets $\overline{S_i}$ can be expressed as a union of simpler sets which involve only inequality constraints similar to the sets $C_i$. Specifically, 
\begin{equation}
    \overline{S_{i}} := \bigcup_{k \in K} \overline{S_{i,k}},
\end{equation}
where the sets comprising the union have the form,
\begin{equation} \label{eq:local_piece}
    \overline{S_{i,k}} := \left\{ \begin{aligned} 
        & F_j(z) = 0, l_j \le z_j \le u_j, &j \in J_1 \\
        & F_j(z) \ge 0, l_j = z_j, &j \in J_{2} \\
        & F_j(z) \le 0, z_j = u_j, &j \in J_{3}
        \end{aligned} \right\},
\end{equation}
for some appropriately sized index sets $J_1$,$J_2$, $J_3$. Interpreting the set $\overline{S_i}$ via this union comes from enumerating the possible ways the conditions \cref{eq:comp_def} can be satisfied. Using this definition, the minimization in \cref{eq:stationary_graph_bilevel} becomes:
\begin{equation} \label{eq:simple_union}
    \begin{aligned}
        \min_{x} \ &f(x) \quad
        \mathrm{s.t.} \ \ & x \in \bigcup_{k\in K} D_k.
    \end{aligned}
\end{equation}
The following lemma allows us to reason about local optima of \cref{eq:simple_union}. 
\begin{lemma} \label{lem:1}
    \emph{Let the sets $D_k$ appearing in \cref{eq:simple_union} be closed sets. Then a point $x^*$ is a local optimum of \cref{eq:simple_union} if and only if $x^*$ is a local optimum of all problems:} 
    \begin{equation} \label{eq:single_piece}
        \min_{x \in D_k} f(x), \text{for all $k$ s.t. $x^* \in D_k$}.
    \end{equation} 
\end{lemma}
\begin{proof}
    Let $D := \bigcup_{k\in K} D_k$. Furthermore, let
    \begin{equation}
        \gamma_D(x) := \{k \in K : x \in D_k \}.
    \end{equation}
    By definition, $x^*$ is a local optimum of  \cref{eq:simple_union} if and only if $x^* \in D$. There exists some $\epsilon > 0$ such that,
    \begin{equation} \label{eq:s1}
        f(x^*) \le f(x) \ \forall x \in D : \| x - x^*\| < \epsilon.
    \end{equation}
    Since the sets $D_k$ are all closed, their complements are open, and therefore for sufficiently small choice of $\epsilon$, 
    \begin{equation}
        x \notin D_k \ \forall k \notin \gamma_D(x^*).
    \end{equation}
    Using such a choice of $\epsilon$, the definition in (\ref{eq:s1}) can be rewritten as 
     \begin{equation} \label{eq:s2}
        f(x^*) \le f(x) \ \forall x \in D_k : \| x - x^*\| < \epsilon, k \in \gamma_D(x^*).
    \end{equation}
    Since $x^* \in D_k \ \forall k \in \gamma_D(x^*)$, this is equivalent to $x^*$ being a local optimum for each of the problems (\ref{eq:single_piece}). 
\end{proof}

The intuition behind this lemma is simple: If there are no local regions to descend within the unioned feasible set, then no individual set in the union would offer a descent direction either. The converse is also true.
This result inspires yet another set related to the bilevel solution graph, now only using a single component $\overline{S_{-i,k}}$ of the larger set of stationary points: 
\begin{equation}
    \label{eq:stationary_graph_bilevel2}
    \overline{B_{i,k}} = \left\{ \begin{aligned} 
        x^*: \exists \lambda^*, (x^*,\lambda^*) \in &\argmin_{x,\lambda}  && f_i\left(x\right) \\
        & \ \ \ \ \text{s.t.} && x \in C_i \\
        & && \left(x,\lambda\right) \in \overline{S_{-i,k}}
        \end{aligned} \right\}
\end{equation}

Leveraging \cref{lem:1} and the sets \cref{eq:stationary_graph_bilevel2}, an algorithm to check if a candidate point $x^*$ is a bilevel equilibrium can be derived.  First, we identify a pair $(x, \lambda) \in S_{-i} \subset \overline{S_{-i}}$, by solving the follower's optimization problem. Then we enumerate all local regions \cref{eq:local_piece}, i.e. those such that $(x, \lambda) \in \overline{S_{-i,k}}$. For each of these regions and associated indices $k$, we check to see if $(x, \lambda)$ is an element of $\overline{B_{i,k}}$. If the point belongs to all sets, then it must be an element of $\overline{B_i}$. This procedure is developed into an algorithm for computing bilevel equilibrium points in \cref{alg:cap}.

Solving for bilevel optima is generally NP-hard due to the combinatorial nature of the complementarity conditions which define the follower's solution graph. Nevertheless, in the games studied in this work, bilevel optima can often be computed in fractions of a second using \cref{alg:cap}. This depends on the problem data and size, and for difficult instances, solve times can be much longer, or \cref{alg:cap} can fail. While efficiency and reliability is important, it is not a critical component of this work. Rather, we are investigating the advantage posed by utilizing a bilevel formulation, which might inspire the development of more efficient and reliable techniques.
\begin{algorithm}[b]
\caption{Computing Bilevel Equilibria (for Player $1$)}\label{alg:cap}
\begin{algorithmic}
\Require Initial solution guess $x^*$
\While{not solved}
    \State $(x^*_2, \lambda) \gets$ sol. to player 2's optimization for $x_1^*$; $K \gets$ index set s.t. $(x^*,\lambda) \in \overline{S_{2,k}} \ \forall k \in K$; all agree $\gets$ true
    \For{$k \in K$}
        \If{$(x^*,\lambda) \notin \overline{B_{i,k}}$}
            \State all agree $\gets$ false; $(x^*, \lambda) \gets$ solution in $\overline{B_{i,k}}$
        \EndIf
    \EndFor
    \If{all agree}
        \State Return $x^*$
    \EndIf
\EndWhile
\end{algorithmic}
\end{algorithm}

\section{Two-player Racing Game}
\label{sec:racing_model}

Racing is a speed competition in any physical domains where a human-controlled or autonomous vehicle can move. To stay competitive, players must operate their vehicles near the dynamic limits and often maneuver close to other vehicles. This section describes a simplified nonlinear model of a two-player racing with generic vehicles on the plane.

Our racing game features two simple craft on a constant-width, lane-free track with bends. %
We assume players control their tangential acceleration and angular velocity (rate of heading). Most importantly, players must avoid collisions and stay within track limits. Acceleration and collision constraints drive competitive interactions, because they depend on the opposing player's current position and actions, which the ego player cannot directly control. There is no cooperation between the players. Our racing game and approach is fully model-driven, and does not require data-driven prior knowledge, for example unlike other recent non-cooperative dynamic game solvers that require a reference policy \cite{lidard_blending_2024}.

Players select their control inputs for the next $n_T$ horizon time steps at each simulation step by solving a dynamic game with some assumed solution concept. The solution concept of the two players need not be consistent. We assume players are fully aware of the opponent's state, but not the solution concept of the opponent. We assume that players cannot switch solution concepts during play. 

\subsection{Vehicle Dynamics Model}
\label{subsec:dynamics}

The vehicle dynamics involve first-order differential equations describing craft motion, with players controlling their own tangential acceleration and angular velocity. We add a linear forcing term for velocity-dependent drag. Let $p_i\in\mathbb{R}^2$ refer to the $i$\textsuperscript{th} player's position in the Cartesian coordinate system. We use superscripts to refer to the indices of vectors. Let $y_i$ be the state and $u_i$ be the control input for the $i$\textsuperscript{th} player,
\begin{equation*}
\begin{aligned}
y_i &= 
\begin{bmatrix}
p_i^\text{lat}&
p_i^\text{long}&
v_i &
\theta_i \end{bmatrix}^\intercal \in \mathbb{R}^{4},\quad
u_i = \begin{bmatrix}
\tau \\
\omega\end{bmatrix}\in \mathbb{R}^2,
\end{aligned}
\end{equation*}
$v_i$ is the tangential velocity, $\theta_i$ is the heading, $\tau$ is tangential thrust (or braking) acceleration and $\omega$ is angular velocity.

We use backwards (implicit) Euler integration for discretization, then at $k$\textsuperscript{th} time step,
\begin{equation*}
\begin{aligned}
y_i[k+1] =
\begin{bmatrix}  
p_i^\text{lat}[k] + \Delta t\; v_i[k+1] \;\cos \theta_i[k+1] \\
p_i^\text{long}[k] + \Delta t \; v_i[k+1] \;\sin\theta_i[k+1] \\
v_i[k+1] \\
\theta_i[k+1]
\end{bmatrix}.
\end{aligned}
\end{equation*}
Future tangential velocity and heading is determined by unity coefficients and, $c_\text{drag}$ is the linear drag coefficient,
\begin{equation*}
\begin{aligned}  
v_i[k+1] = \tau[k] - c_\text{drag} v_i[k], \quad
\theta_i[k+1] = \omega[k].
\end{aligned}
\end{equation*}
The joint set of states of the game is ${y \in \mathbb{R}^8 := [y_1^\intercal, y_2^\intercal]^\intercal}$. 

The private decision variables of $i$\textsuperscript{th} player at $k$\textsuperscript{th} simulation step is the player's state and control variables for the next $n_T$ time steps into the future:
\begin{equation*}
\begin{aligned}
x_i[k] &= \begin{bmatrix}  
\begin{bmatrix} y_i[k+1]^\intercal, \cdots, y_i[k+n_T]^\intercal  \end{bmatrix}^\intercal \\
\begin{bmatrix} u_i[k+1]^\intercal, \cdots, u_i[k+n_T]^\intercal  \end{bmatrix}^\intercal
\end{bmatrix} \in \mathbb{R}^{6n_T}.
\end{aligned}
\end{equation*}
The total number of decision variables is $n=12n_T$. We omit the time step indexing in square brackets when it's clear.

\subsection{Objective Functions and Constraints}
\label{subsec:obj_func}

The players' objective functions, defined for player $i$ is,
\begin{equation*}
\begin{aligned}
f_i(x) = 
\sum_{t=k+1}^{k+n_T} \alpha_1^2\,  (d_i^2 - r_i^2)^2 +\alpha_2\, u_i^\intercal u_i+  \beta (v_{-i}^\text{long} - v_{i}^\text{long}),
\end{aligned}
\end{equation*}
and it consists of three components integrated over the next $n_T$ time steps: (i) A penalty for being away from the road center, (ii) a quadratic penalty for control effort, (iii) and a linear penalty for the relative advantage of the opponent's longitudinal velocity. The velocity penalty term is positive when the opponent's velocity is greater than the ego's velocity, and vice versa. 
We use shorthand $d_i$ for the distance measured from the center of curvature of the track (three center-line checkpoints defines it), and $v_{-i}^\text{long}$ is the longitudinal velocity, 
\begin{equation*}
\begin{gathered}
d_i = (p_i - C_i)^\intercal(p_i - C_i), \\
v_{-i}^\text{long} = v_{-i}\sin(\theta_{-i}),\quad
v_{i}^\text{long} = v_{i} \sin(\theta_i),
\end{gathered}
\end{equation*}
$C_i\in\mathbb{R}^2$ is the position of the center of curvature and $r_i>0$ is the radius of curvature, discussed in more detail in the next subsection. The only part of the cost that depends on the opponent's actions is the relative velocity component, therefore this is the only component that encourages competition.

Each player has $4 n_T$ equality constraints for dynamics, $n_T$ inequality constraints for collision $g_i^\text{col}(x_1, x_2)$, $2 n_T$ inequality constraints for the track limits, $3n_T$ inequality constraints on tangential velocity $v_i$ to limit reverse speeds, and to ensure $\theta_i\in[-\pi/2, \pi/2]$, and $4n_T$ inequalities to ensure the tangential acceleration and rate of heading (the control inputs) are constrained. Responsibility of avoiding collisions is determined by the relative longitudinal position of the players. Lastly, the tangential acceleration constraint is a function of the opponent's position, intended encourage competitive interactions as a proxy for drafting.

\subsection{Track Model and Constraints}

The track is defined as a list of lateral offsets specifying the center-line at different longitudinal positions. The track is approximated by the center $C_i$ and the radius $r_i$ of the circle that passes through the three nearest track center-line checkpoints for $i$\textsuperscript{th} player. A limitation of our track model is its inability to handle 180-degree turns, so we restrict heading to $[-\pi/2, \pi/2]$. Another drawback is that three consecutive checkpoints must not be collinear, which we avoid by adding small perturbations to the straight sections. We define the track constraints using circle equations and track width $w_\text{track}$,
\begin{equation*}
\begin{aligned}
(p_i - C_i)^\intercal(p_i - C_i) - \left(r_i - \frac{w_\text{track}}{2}\right)^2 &\geq 0, \\
\left(r_i + \frac{w_\text{track}}{2}\right)^2 - (p_i - C_i)^\intercal(p_i - C_i)  &\geq 0.
\end{aligned}
\end{equation*}

\subsection{Collision Avoidance and Responsibility} \label{sec:responsibility}

The collision avoidance constraint ensures a minimum radial distance between players. 
Collision avoidance depends on and benefits both players, however the responsibility is not shared.  %
Collision avoidance responsibility is determined by the players' relative longitudinal positions. The responsibility of collision avoidance falls to the player that is behind in terms of longitudinal position. The responsibility is modeled using a custom responsibility function with a transition region for when the players are approximately side by side,
\begin{equation*}
\begin{aligned}
g_{i, k}^\text{col}(x) &= \begin{matrix}\lVert p_1 - p_2\rVert^2 - r_\text{col}^2 - \ell_i(p_2^\text{long} -  p_1^\text{long})\end{matrix} \geq 0, \\
\ell_1(h) &= \frac{1}{1 + e^{a h + b} } - \frac{1}{1 + e^b}, \\
\ell_2(h) &= \frac{1}{1 + e^{-a h + b} } - \frac{1}{1 + e^b},
\end{aligned}
\end{equation*}
where $\lVert\cdot\rVert$ is the Euclidean norm, $r_\text{col}$ is the minimum desired radial distance between the players, and $a$ and $b$ adjust the slope of the sigmoid function near zero. 

We use positional constraints to set the lower bound of the collision constraint.
You can see how $l_1$ and $l_2$ depend on the relative positions in \cref{fig:ells} for $a=a=5$, $b=4.5$. Because collision constraints include both players' actions, the strategy-specific bounds $l_1$ and $l_2$ do not fundamentally change the constraints' nature for equilibrium computation and existence.
Collision constraint is enforced for every time step $k\in\{t+1, \cdots, t+n_T\}$, $g_{i}= \begin{bmatrix} g_{i, t}^\intercal, \cdots, g_{i, t+n_T}^\intercal  \end{bmatrix}^\intercal$.  When $p_2^\text{long} -  p_1^\text{long}\geq 0$ (P2 is ahead), then $\ell_1(h)\geq 0$, $\ell_2(h)\leq 0$ and so, $g_2^\text{col}$ is relieved, and vice versa. 
We add a small buffer to the minimum collision distance, but terminate simulations only if players violate absolute collision distance or other track violations.

\begin{figure}[t]
\centering
\includegraphics[width=0.8\columnwidth]{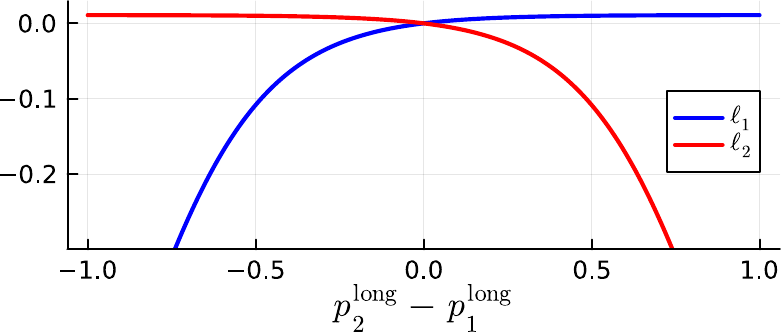}
\caption{$\ell_1$ and $\ell_2$ with respect to relative longitudinal position of the players, bounded above by a small positive value.} %
\label{fig:ells}
\end{figure}

\subsection{Drafting for Acceleration Limits}

\begin{figure}[b]
\centering
\includegraphics[width=0.8\columnwidth]{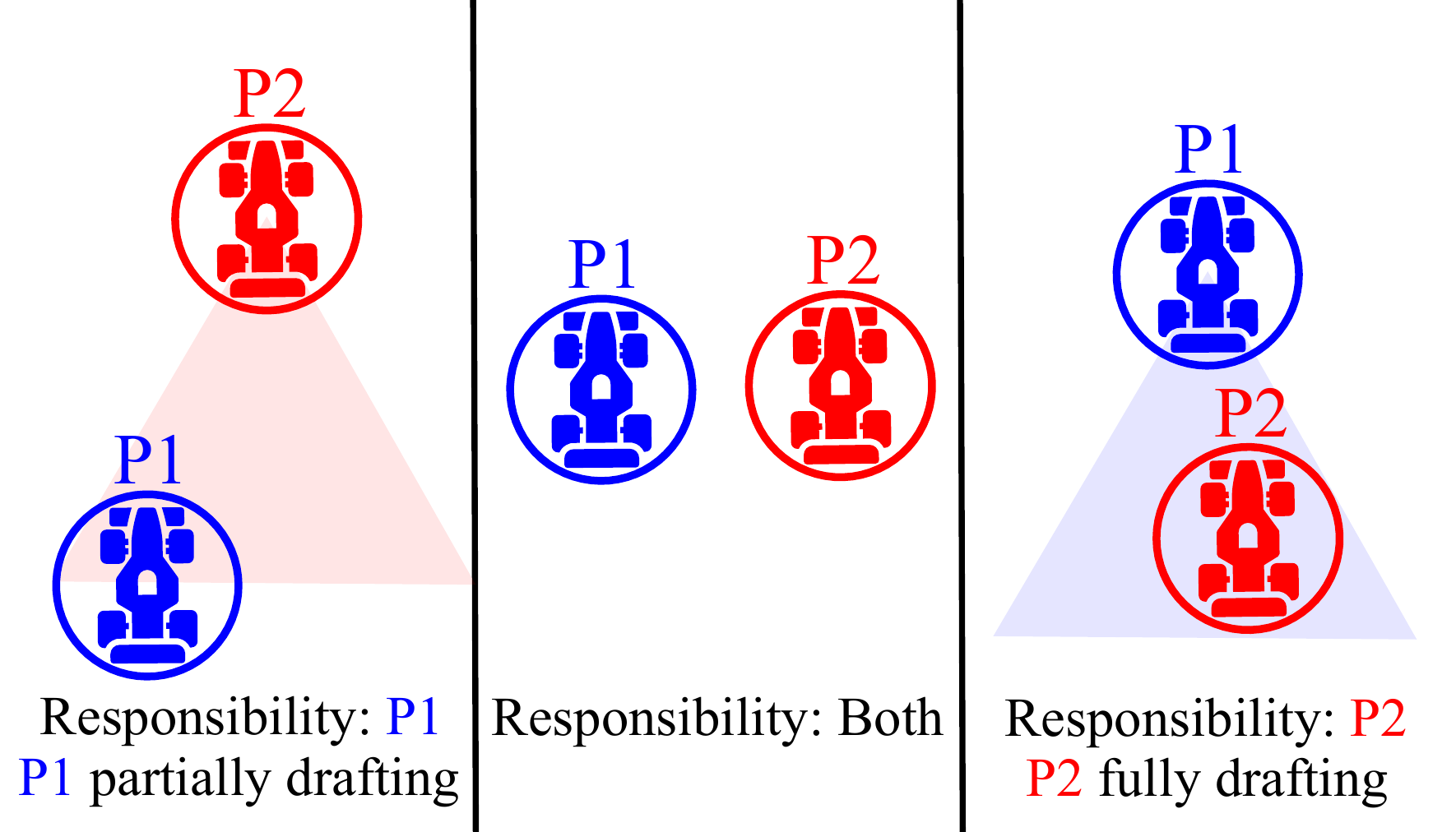}
\caption{Triangular drafting region and collision avoidance responsibility based on relative positions of players.}
\label{fig:responsibility}
\end{figure}

We use a linear drag as a function velocity, requiring constant control effort to maintain speed. To encourage dynamic competition, we apply positional constraints on players' acceleration limits. As an approximation of the aerodynamics of racing, we define a smooth triangular region behind each player to boost the opponent's acceleration limits, as shown in \cref{fig:responsibility}. The constraints on the control input can be summarized as $\tau_i \geq \tau_\text{min}$, $\omega_\text{max} \geq \omega_i$, $\omega_i \geq-\omega_\text{max}$, and $\tau^\text{draft}(p_{i}, p_{-i}) \geq \tau_i$,
where the upper limit of tangential acceleration is defined in $\tau_\text{max}^\text{draft} \geq \tau^\text{draft}(p_{i}, p_{-i}) \geq \tau_\text{max}^\text{nom}$, which takes on a greater value when the ego is in the smooth triangular area behind the opponent, as in \cref{fig:responsibility}. 
When the ego is outside the drafting zone, then $\tau^\text{draft}(p_{i}, p_{-i}) =  \tau_\text{max}^\text{nom}$. Parameters $\ell_\text{draft}, w_\text{draft}$ define the size of the drafting region. The drafting encourages close interactions, and larger (negative) braking accelerations enable defensive strategies.

\section{Simulations and Discussion}
\label{sec:sim_experiments}

We compare 16 different competition types as a combination of different information patterns using our two-player racing model in a large-scale randomized simulation trial. Players choose from the following strategies: (1) Single-player problem, (2) Nash equilibrium, (3) leader in bilevel equilibrium, (4) follower in bilevel equilibrium. Players do not know about their opponents' choice of strategy, so the players disagree on the information pattern in 12 of the 16 competition types. The single-player game involves optimizing the ego problem \cref{eq:solution_graph} when using a fixed trajectory for opponent, generated by predicting forward with constant velocity and heading from the initial position of the opponent, i.e. ``constant-velocity Model-Predictive Control''. Due to the symmetry in the information patterns, we only simulate 10 competitions for each configuration.

Even in the absence of disagreement, the trajectory games are imperfect models of the interactions. Both players continually solve finite-horizon equilibria but only execute the first step of the strategy before re-solving, unlike in finite-horizon \emph{feedback} games \cite{laine2023computation}. Players' lack of knowledge about each others' strategies could be modeled as imperfect information games, where a belief distribution over possible opponent strategies is used for optimization \cite{peters2024contingency}. However, repeated and imperfect information games are more complex and computationally demanding. Our approach is tractable while still yielding useful strategies, so we limit our exploration to it for this work.

Bilevel problems are challenging due to their sensitivity to solver initialization, which can lead to getting stuck in poor local equilibria or fail to find an equilibrium. We use a fallback strategy when we are solving for the bilevel equilibrium. We initialize the bilevel solver using the Nash equilibrium solution, in the event that the solver fails, we fall back to using the single-player solution as our initialization. 
If all else fails, we revert to an uncontrolled simulation step, hoping for improvement in subsequent iterations. At the start of every simulation step, we check all constraints. We terminate the simulation if either players' initial state violate track boundaries or collision constraints. 

We sampled 1000 random feasible initial conditions, and a random phase that corresponds to each initial condition to start from a random location on an infinitely repeating track pattern that contains straight sections and turns. We simulated for 10 competition types, for 25 simulation time steps, using $n_T=10$ number computation horizon time steps, for each initial condition. %
For the initial velocities, the ego player's tangential velocity is sampled from a uniform range between $[1.5, 3]\;$ m/s, while the opponent's tangential velocity is selected using another uniformly distributed offset between $[0,1.5]$ m/s with respect to the ego player's velocity, with zero initial heading (vertical).  For all experiments, we use the following parameters: $\alpha_1=0.001$, $\alpha_2=0.0001$,  $\beta=0.1$, $n_T=10$, $\Delta t=\qty{1e-1}{\second}$, $c_\text{drag}=\qty{.1}{\second^{-1}}$, $r_\text{col}=\qty{1.0}{\meter}$, $\tau_\text{max}^\text{nom}=\qty{1.0}{\meter\second^{-2}}$, $\tau_\text{max}^\text{draft}=\qty{3}{\meter\second^{-2}}$, $\tau_\text{min}=\qty{-3.0}{\meter\second^{-2}}$, $\omega_\text{max} = \qty{3}{\second^{-2}}$, $w_\text{track} = \qty{4.0}{\meter}$, $w_\text{draft}=\qty{5.0}{\meter}$, $\ell_\text{draft} = \qty{5.0}{\meter}$.

\subsection{Robustness}

Our robustness metric is the mean number of steps before (i) a player is forced out of track, or (ii) players collide. The mean number of steps for all competition types are presented in \cref{tab:steps}, with \%95 confidence interval (CI). Player 1 (P1) is the column player, while Player 2 (P2) is the row player. 

Assume P1 is the ego player and the opponent is P2, and P2 chooses their strategy randomly from an equal probability distribution, then the average of row represents expected robustness per ego ego strategy. We find playing as leader in bilevel equilibrium strategy leads to reduced robustness on average, as boldfaced in \cref{tab:steps}. 
The average robustness difference between SP, Nash and Follower strategies are statistically insignificant, which are all more robust than the leader strategy.
Playing as a leader increases aggression, which results in reduced safety. The leader-leader competition leads to the lowest robustness of any competition type, as boldfaced in \cref{tab:steps}.
We stipulate this is due to the players having a false sense of confidence about their ability to enforce their strategies on their opponent.

\begin{table}[b]
\centering
\caption{Robustness: Mean number of simulation time steps}%
	\begin{tabular}{
     r
     S[table-format=2.1(1)]
     S[table-format=2.1(1)]
     S[table-format=2.1(1)]
     S[table-format=2.1(1)]
     }
	\toprule
	\diagbox{P2}{P1} & {SP (S)} & {Nash (N)} & {Leader (L)} & {Follower (F)} \\
	\midrule
	SP (S) & 20.00(47) & 19.76(48) & 18.80(50) & 20.08(48)  \\
	Nash (N)& 19.76(48) & 19.63(48) & 18.12(51) & 19.72(48) \\
    Leader (L) & 18.80(50) & 18.12(51) & \bftabnum 16.76(50) & 19.05(49) \\
    Follower (F) &  20.08(48) & 19.72(48) & 19.05(49) & 20.10(48) \\
    \midrule
    Average & 19.66(44) & 19.30(44) &  \bftabnum 18.18(43) & 19.74(44) \\
 	\bottomrule
	\end{tabular}
\label{tab:steps}
\end{table}

\subsection{Performance}
\label{subsec:performance}

In this section, we compare  performance across competition types using the mean total running cost for the ego player (P1), which reflects what the players are actually optimizing for. We compile the performance of P1 for different competition types in \cref{tab:acosts}. Because P1 and P2 have the same dynamics and constraints, the table for P2 would be the transpose of \cref{tab:acosts}. The values in these tables are not symmetric along the diagonal, because the total cost is not zero-sum due to lane and control penalty components. 

\begin{table}[t]
\centering
\caption{Performance: Mean total running cost for P1 ($\times 100$) (Transpose for P2 cost)} 
	\begin{tabular}{
     r
     S[table-format=2.1(1)]
     S[table-format=2.1(1)]
     S[table-format=2.1(1)]
     S[table-format=2.1(1)]
     }
	\toprule
	\diagbox{P2}{P1} & {SP (S)} & {Nash (N)} & {Leader (L)} & {Follower (F)} \\
	\midrule
	SP (S) & 1.214(69) & 0.30(39) & 0.64(38) & 0.67(37)  \\
	Nash (N)& 2.09(37) & 1.148(67) & 0.30(42) & 0.22(38) \\
    Leader (L) & 1.23(38) & 1.41(42) & 0.625(50) &  0.13(41) \\
    Follower (F) & 1.75(35) & 2.05(37) & 1.80(48) &  \bftabnum 1.200(66) \\
    \midrule
    Average & 1.57(26) & 1.22(14) & .84(15) & \bftabnum .55(27) \\
 	\bottomrule
	\end{tabular}
\label{tab:acosts}
\end{table}

For a P2 opponent who chooses their strategy from a uniform strategy distribution, the average cost in the bottom-most row in \cref{tab:acosts} indicates expected performance for P1 strategies. Interestingly, we find that playing as the follower against an opponent who is the leader provides the lowest cost.
Another way to think about the average relative cost advantage is in terms of the meta  bimatrix game that emerges from \cref{tab:acosts}. This bimatrix game has a unique Nash equilibrium, and the reader can verify that the Follower-Follower competition (boldfaced in \cref{tab:acosts}) is the only Nash equilibrium, 
which is stable for all values within the confidence interval range. 

The bilevel competition leads to increased competitive advantage to the leader, except when the opponent adopts a follower strategy. 
The competitive advantage comes with increased collision risk and therefore reduced robustness, or safety, especially if both players greedily assume that they are the leader, the likelihood of a collision or track boundary violation increases significantly. 
Remarkably, the best choice in terms of performance and robustness is for players to unilaterally assume that they are the follower.

\section{Conclusion} 
\label{sec:conclusion}

In this study we performed a large scale empirical analysis on the benefits felt from computing trajectories via bilevel competition compared to static reasoning. We have shown that the information structure has a significant effect on the competitive performance and safety. We observed that bilevel reasoning can lead to competitive advantage in the physical racing domain, but it is potentially a trade-off with safety if the players do not agree on the competition type. 
When players act greedily and end up in Leader-Leader competition, safety decreases for all. While an always-follower strategy can improve safety, it can hurt competitive performance.

Our results could inform roboticists designing control systems in autonomous multi-agent settings. We show that even when the opponent's strategy is unknown, the choice of ego strategy have performance and safety consequences. If the opponent's strategy is known (for example, by observation over time), the trade-off between safety and performance could the guide selection between different solution concepts to satisfy design goals. An interesting extension to this work is allowing the players to change their strategy during the race, for example to guide an adaptive control scheme.

\bibliographystyle{IEEEtran}
\bibliography{references}

\end{document}


\title{Does bilevel optimization result in more competitive racing behavior?}

\author{Author Names Omitted for Anonymous Review. Paper-ID 389}

\maketitle

\IEEEpeerreviewmaketitle

\section{Supplementary Material}
\begin{proof}[Proof of \cref{lem:1}]
    Let $D := \bigcup_{k\in K} D_k$. Furthermore, let
    \begin{equation}
        \gamma_D(x) := \{k \in K : x \in D_k \}.
    \end{equation}
    By definition, a point $x^*$ is a local optimum for (\ref{eq:simple_union}) if and only if $x^* \in D$, and there exists some $\epsilon > 0$ such that
    \begin{equation} \label{eq:s1}
        f(x^*) \le f(x) \ \forall x \in D : \| x - x^*\| < \epsilon.
    \end{equation}
    Since the sets $D_k$ are all closed, their complements are open, and therefore for sufficiently small choice of $\epsilon$, 
    \begin{equation}
        x \notin D_k \ \forall k \notin \gamma_D(x^*).
    \end{equation}

    Using such a choice of $\epsilon$, the definition in (\ref{eq:s1}) can be rewritten as 
     \begin{equation} \label{eq:s2}
        f(x^*) \le f(x) \ \forall x \in D_k : \| x - x^*\| < \epsilon, k \in \gamma_D(x^*).
    \end{equation}
    Since $x^* \in D_k \ \forall k \in \gamma_D(x^*)$, this is equivalent to $x^*$ being a local optimum for each of the problems (\ref{eq:single_piece}). 
\end{proof}